\documentclass[final, 10pt]{IEEEtran}

\usepackage[dvips]{graphicx}
\usepackage{amssymb,amsmath}
\usepackage{subfig}
\usepackage{colortbl}
\usepackage{epsfig}
\usepackage{algorithm2e}

\graphicspath{{./images/}}

\newtheorem{theorem}{Theorem}[section]
\newtheorem{lemma}{Lemma}[section]

\newtheorem{corollary}{Corollary}[section]
\newtheorem{definition}{Definition}[section]

\newcommand{\ie}{\emph{i.e. }}
\newcommand{\etal}{\textit{et al.}}

\begin{document}

\title{Algebraic Watchdog: Mitigating Misbehavior in Wireless Network Coding}
\author{MinJi Kim\authorrefmark{1}, Muriel M\'{e}dard\authorrefmark{1}, Jo\~{a}o Barros\authorrefmark{2}\vspace*{-.5cm}
\thanks{This work was partially presented at IEEE ISIT 2009 (Seoul, Korea) titled ``An Algebraic Watchdog for Wireless Network Coding'', and at IEEE ITW 2010 (Dublin, Ireland) titled ``A Multi-hop Multi-source Algebraic Watchdog''.}
\thanks{\authorrefmark{1}M. Kim and M. M\'{e}dard (\{minjikim, medard\}@mit.edu) are with the Research Laboratory of Electronics at the Massachusetts Institute of Technology, MA USA. \authorrefmark{2}J. Barros (jbarros@fe.up.pt) is with the Instituto de Telecommunica\c{c}\~{o}es, Faculdade de Engenharia da Universidade do Porto, Portugal.
}
}

\maketitle

\begin{abstract}
We propose a secure scheme for wireless network coding, called the algebraic watchdog. By enabling nodes to detect malicious behaviors probabilistically and use overheard messages to police their downstream neighbors locally, the algebraic watchdog delivers a secure global \emph{self-checking network}.
Unlike traditional Byzantine detection protocols which are \emph{receiver-based}, this protocol gives the senders an active role in checking the node downstream. The key idea is inspired by Marti \etal's \emph{watchdog-pathrater}, which attempts to detect and mitigate the effects of routing misbehavior.

As an initial building block of a such system, we first focus on a two-hop network. We present a graphical model to understand the inference process nodes execute to police their downstream neighbors; as well as to compute, analyze, and approximate the probabilities of misdetection and false detection. In addition, we present an algebraic analysis of the performance using an hypothesis testing framework that provides exact formulae for probabilities of false detection and misdetection.

We then extend the algebraic watchdog to a more general network setting, and propose a protocol in which we can establish \emph{trust} in coded systems in a distributed manner. We develop a graphical model to detect the presence of an adversarial node downstream within a general multi-hop network. The structure of the graphical model (a trellis) lends itself to well-known algorithms, such as the Viterbi algorithm, which can compute the probabilities of misdetection and false detection. We show analytically that as long as the min-cut is not dominated by the Byzantine adversaries, upstream nodes can monitor downstream neighbors and allow reliable communication with certain probability. Finally, we present simulation results that support our analysis. 
\end{abstract}


\section{Introduction}\label{sec:Introduction}
There have been numerous contributions to secure wireless networks, including key management, secure routing, Byzantine detection, and various protocol designs (for a general survey on this topic, see \cite{hubaux}\cite{security01}\cite{security02}\cite{security03}\cite{perlman}\cite{liskov}\cite{Lamport}\cite{papadimitratos}). Countering these types of threats is particularly important in military communications and networking, which are highly dynamic in nature and must not fail when adversaries succeed in compromising some of the nodes in the network. We consider the problem of Byzantine detection. The traditional approach is \emph{receiver-based} -- \ie the receiver of the corrupted data detects the presence of an upstream adversary. However, this detection may come too late as the adversary is partially successful in disrupting the network (even if it is detected). It has wasted network bandwidth, while the source is still unaware of the need for retransmission.

Reference \cite{marti} introduces a protocol for routing wireless networks, called the \emph{watchdog and pathrater}, in which upstream nodes police their downstream neighbors using \emph{promiscuous monitoring}. Promiscuous monitoring means that if a node $v$ is within range of a node $v'$, it can overhear communication to and from $v'$ even if those communication do not directly involve $v$. This scheme successfully detects adversaries and removes misbehaving nodes from the network by dynamically adjusting the routing paths. However, the protocol requires a significant overhead (12\% to 24\%) owing to increased control traffic and numerous cryptographic messages \cite{marti}.

Our goal is to design and analyze a watchdog-inspired protocol for wireless networks using network coding. We propose a new scheme called the \emph{algebraic watchdog}, in which nodes can detect malicious behaviors probabilistically by taking advantage of the broadcast nature of the wireless medium. Although we focus on detecting malicious or misbehaving nodes, the same approach can be applied to faulty or failing nodes. Our ultimate goal is a robust \emph{self-checking network}. The key difference between the our work \cite{algebraicwatchdog} and that of \cite{marti} is that we allow network coding. Network coding \cite{ahlswede}\cite{algebraic} is advantageous as it not only increases throughput and robustness against failures and erasures but also it is resilient in dynamic/unstable networks where state information may change rapidly or may be hard to obtain.

The key challenge in algebraic watchdog is that, by incorporating network coding, we can no longer recognize packets individually. In \cite{marti}, a node $v$ can monitor its downstream neighbor $v'$ by checking that the packet transmitted by $v'$ is a copy of what $v$ transmitted to $v'$. However, with network coding, this is no longer possible as transmitted packets are a function of the received packets. Furthermore, $v$ may not have full information regarding the packets received at $v'$; thus, node $v$ is faced with the challenge of inferring the packets received at $v'$ and ensuring that $v'$ is transmitting a valid function of the received packets. We note that \cite{whenmeets} combines source coding with watchdog; thus, \cite{whenmeets} does not face the same problem as the algebraic watchdog.

The paper is organized as follows. In Section \ref{sec:intuition}, we briefly discuss the intuition behind algebraic watchdog. In Section \ref{sec:watchdog_background}, we present the background and related material. In Section \ref{sec:watchdog_problemstatement}, we introduce our problem statement and network model. In Section \ref{sec:watchdog_example}, we analyze the protocol for a simple two-hop network, first algebraically in Section \ref{sec:watchdog_algebraicapproach} and then graphically in Section \ref{sec:watchdog_graphicalmodel}. In Section \ref{sec:watchdog_multisource}, we extend the analysis for algebraic watchdog to a more general two-hop network, and in Section \ref{sec:watchdog_multihop}, we present an algebraic watchdog protocol for a multi-hop network. We present simulation results in Section \ref{sec:watchdog_simulation}, which confirm our analysis and show that an adversary within the network can be detected probabilistically by upstream nodes. In Section \ref{sec:watchdog_conclusions}, we summarize our contribution and discuss some future work.

\section{Intuition}\label{sec:intuition}
Consider a network in which the sources are well-behaving (If the sources are malicious, there is no ``uncorrupted'' information flow to protect). In such a case, the sources can monitor their downstream neighbors as shown in Figure \ref{fig:watchdog_prop}. Assume that nodes $v_1, v_2, v_3,$ and $v_4$ are sources. Nodes in $S_0 = \{v_1, v_2, v_3\}$ can monitor $v_5$ collectively or independently. In addition, $v_3$ and $v_4$ can monitor $v_6$. This enforces $v_5$ and $v_6$ to send valid information. Note that we do not make any assumption on whether $v_5$ and $v_6$ are malicious or not -- they are forced to send valid information regardless of their true nature.

If it is the case that $v_5$ and $v_6$ are well-behaving, then we can employ the same scheme at $v_5$ or $v_6$ to check $v_7$'s behavior. Thus, propagating \emph{trust} within the network. Now, what if $v_5$ or $v_6$ are malicious? If both $v_5$ and $v_6$ are malicious, all flows to $v_7$ are controlled by malicious nodes -- \ie flows through $v_7$ are completely compromised. Therefore, even if $v_7$ is well-behaving, there is nothing that $v_7$ or $v_1, v_2, v_3, v_4$ can do to protect the flow through $v_7$. The only solution in this case would be to physically remove $v_5$ and $v_6$ from the network or to construct a new path to $v_7$.

The intuition is that as long as the min-cut to any node is not dominated by malicious nodes, then the remaining well-behaving nodes can check its neighborhood and enforce that the information flow is delivered correctly to the destination. For example, assume that only $v_6$ is malicious and $v_5$ is well-behaving in Figure \ref{fig:watchdog_prop}. Since $v_3$ and $v_4$ monitor $v_6$, we know that despite $v_6$ being malicious, $v_6$ is forced to send valid information. Then, $v_7$ receives two valid information flows, which it is now responsible of forwarding. If $v_7$ is well-behaving, we do not have any problem. If $v_7$ is malicious, it may wish to inject errors to the information flow. In this case, $v_7$ is only liable to $v_5$; but it is liable to at least one well-behaving node $v_5$. Thus, it is not completely free to inject any error it chooses; it has to ensure that $v_5$ cannot detect its misbehavior, which may be difficult to accomplish.

In this paper, we show that this is indeed the case. We first start by studying a two-hop network, which would be equivalent to focusing on the operations performed by nodes in $S_0$ to check $v_5$. Then, we discuss how we can propagate this two-hop policing strategy to a multi-hop scenario.

\begin{figure}
\begin{center}
\includegraphics[width=0.33\textwidth]{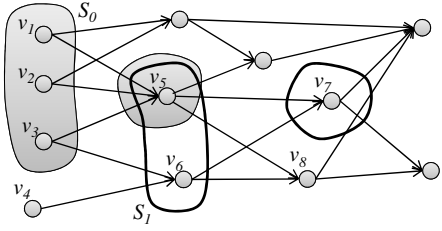}
\end{center}\vspace*{-.2cm}
\caption{An example network.}\label{fig:watchdog_prop}\vspace*{-.4cm}
\end{figure}

\section{Background}\label{sec:watchdog_background}
\subsection{Secure Network Coding}\label{sec:watchdog_securenc}

Network coding, first introduced in \cite{ahlswede}, allows algebraic mixing of information in the intermediate nodes. This mixing has been shown to have numerous performance benefits. It is known that network coding maximizes throughput for multicast \cite{ahlswede} and increases robustness against failures \cite{algebraic} and erasures \cite{reliable}. However, a major concern for network coded systems is their vulnerability to Byzantine adversaries. A single corrupted packet generated by a Byzantine adversary can contaminate all the information to a destination, and propagate to other destinations quickly. For example, in random linear network coding \cite{reliable}, one corrupted packet in a generation (\ie a fixed set of packets) can prevent a receiver from decoding any data from that generation even if all the other packets it has received are valid.

There are several papers that attempt to address this problem. One approach is to correct the errors injected by the Byzantine adversaries using \emph{network error correction} \cite{errorcorrection}. Reference \cite{errorcorrection} bounds the maximum achievable rate in an adversarial setting, and generalizes the Hamming, Gilbert-Varshamov, and Singleton bounds. Jaggi \etal\ \cite{resilient} propose a distributed, rate-optimal, network coding scheme for multicast network that is resilient in the presence of Byzantine adversaries for sufficiently large field and packet size. Reference \cite{subspace} generalizes \cite{resilient} to provide correction guarantees against adversarial errors for any given field and packet size. In \cite{milcom}, Kim \etal\ compare the cost and benefit associated with these Byzantine detection schemes in terms of transmitted bits by allowing nodes to employ the detection schemes to drop polluted data.


\subsection{Secure Routing Protocol: Watchdog and Pathrater}

The problem of securing networks in the presence of Byzantine adversaries has been studied extensively, e.g. \cite{perlman}\cite{liskov}\cite{Lamport}\cite{papadimitratos}. The \emph{watchdog and pathrater} \cite{marti} are two extensions to the Dynamic Source Routing \cite{dsr} protocol that attempt to detect and mitigate the effects of routing misbehavior. The watchdog detects misbehavior based on promiscuous monitoring of the transmissions of the downstream node to confirm if this relay correctly forwards the packets it receives. If a node bound to forward a packet fails to do so after a certain period of time, the watchdog increments a failure rating for that node and a node is deemed to be misbehaving when this failure rating exceeds a certain threshold. The pathrater then uses the gathered information to determine the best possible routes by avoiding misbehaving nodes. This mechanism, which does not punish these nodes (it actually relieves them from forwarding operations), provides an increase in the throughput of networks with misbehaving nodes \cite{marti}.

\subsection{Hypothesis Testing}\label{sec:watchdog_hypothesis}

Hypothesis testing is a method of deciding which of the two hypotheses, denoted $H_0$ and $H_1$, is true, given an observation denoted as $U$. In this paper, $H_0$ is the hypothesis that $v$ is well-behaving, $H_1$ is that $v$ is malicious, and $U$ is the information gathered from overhearing. The observation $U$ is distributed differently depending whether $H_0$ or $H_1$ is true, and these distributions are denoted as $P_{U|H_0}$ and $P_{U|H_1}$ respectively.

An algorithm is used to choose between the hypotheses given the observation $U$. There are two types of error associated with the decision process:
\begin{itemize}
\item{\it{Type 1 error, False detection}}: Accepting $H_1$ when $H_0$ is true (\ie considering a well-behaving $v$ to be malicious), and the probability of this event is denoted $\gamma$.
\item{\it{Type 2 error, Misdetection}}: Accepting $H_0$ when $H_1$ is true (\ie considering a malicious $v$ to be well-behaving), and the probability of this event is denoted $\beta$.
\end{itemize}
The Neyman-Pearson theorem gives the optimal decision rule that given the maximal tolerable $\beta$, we can minimize $\gamma$ by accepting hypothesis $H_0$ if and only if $\log \frac{P_{U|H_0}}{P_{U|H_1}} \geq t$ for some threshold $t$ dependant on $\gamma$. For more thorough survey on hypothesis testing in the context of authentication, see \cite{hypothesis}.

\section{Problem Statement}\label{sec:watchdog_problemstatement}
We shall use elements from a field, and their bit-representation. We use the same character in italic font (\ie $x$) for the field element, and in bold font (\ie $\mathbf{x}$) for the bit-representation. We use underscore bold font (\ie $\mathbf{\underline{x}}$) for vectors. For arithmetic operations in the field, we shall use the conventional notation (\ie $+, -, \cdot
$). For bit-operation, we shall use $\oplus$ for addition, and $\otimes$ for multiplication.

We also require polynomial hash functions defined as follows (for a more detailed discussion on this topic, see \cite{hash}).
\begin{definition}[\textit{\textbf{Polynomial hash functions}}] For a finite field $\mathbf{F}$ and $d \geq 1$, the class of polynomial hash functions on $\mathbf{F}$ is defined as follows:
\[
\mathcal{H}^d(\mathbf{F}) = \{h_a | a = \langle a_0, ..., a_d \rangle \in \mathbf{F}^{d+1}\},
\]
where $h_a(x) = \sum_{i=0}^d a_i x^i$ for $x\in \mathbf{F}$.
\end{definition}

We model a wireless network with a hypergraph $G = (V, E_1, E_2)$, where $V$ is the set of the nodes in the network, $E_1$ is the set of hyperedges representing the connectivity (wireless links), and $E_2$ is the set of hyperedges representing the interference. We use the hypergraph to capture the broadcast nature of the wireless medium. If $(v_1, v_2) \in E_1$ and $(v_1,v_3)\in E_2$ where $v_1, v_2, v_3\in V$, then there is an intended transmission from $v_1$ to $v_2$, and $v_3$ can overhear this transmission (possibly incorrectly). There is a certain transition probability associated with the interference channels known to the nodes, and we model them with binary channels, $BSC(p_{ij})$ for $(v_i, v_j) \in E_2$.

A node $v_i\in V$ transmits coded information $x_i$ by transmitting a packet $\mathbf{\underline{p_i}}$, where $\mathbf{\underline{p_i}}=[\mathbf{a_i}, \mathbf{h_{I_i}}, \mathbf{h_{x_i}}, \mathbf{x_i}]$ is a $\{0,1\}$-vector. A valid packet $\mathbf{\underline{p_i}}$ is defined as below:
\begin{itemize}
\item $\mathbf{a_i}$ corresponds to the coding coefficients $\alpha_j$, $j \in I_i$, where $I_i \subseteq V$ is the set of nodes adjacent to $v_i$ in $E_1$,
\item $\mathbf{h_{I_i}}$ corresponds to the hash $h(x_j)$, $v_j \in I_i$ where $h(\cdot)$ is a $\delta$-bit polynomial hash function,
\item $\mathbf{h_{x_i}}$ corresponds to the polynomial hash $h(x_i)$,
\item $\mathbf{x_i}$ is the $n$-bit representation of $x_i = \sum_{j \in I} \alpha_j x_j \in (\mathbf{F}_{2^n})$.
\end{itemize}
Figure \ref{fig:watchdog_packet} illustrates the structure of a valid packet. For simplicity, we assume the payload $\mathbf{x_i}$ to be a single symbol. We design and analyze our protocol for a single symbol. However, the protocol applies (and therefore, the analysis) to packets with multiple symbols by applying the protocol on each symbol separately.

\begin{figure}[tbp]
\begin{center}\vspace*{.2cm}
\includegraphics[width=0.31\textwidth]{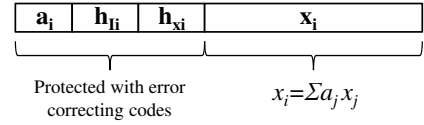}
\end{center}\vspace*{-.2cm}
\caption{A valid packet $\mathbf{\underline{p_i}}$ sent by well-behaving $v_i$.}\label{fig:watchdog_packet}\vspace*{-.3cm}
\end{figure}

The payload $\mathbf{x_i}$ is coded with a $(n, k_i)$-code $\mathcal{C}_i$ with minimum distance $d_i$. Code $\mathcal{C}_i$ is an error-correcting code of rate $R_i = \frac{k_i}{n} = 1-\frac{d_i}{n}$, and is tailored for the forward communication. For instance, $v_{1}$ uses code $\mathcal{C}_1$, chosen appropriately for the channel $(v_1, v_j)\in E_1$, to transmit the payload $\mathbf{x_1}$.

We assume that the payload $\mathbf{x_i}$ is $n$-bits, and the hash $h(\cdot)$ is $\delta$-bits. We assume that the hash function used, $h(\cdot)$, is known to all nodes, including the adversary. In addition, we assume that $\mathbf{a_i}$, $\mathbf{h_{I_i}}$ and $\mathbf{h_{x_i}}$ are part of the header information, and are sufficiently coded to allow the nodes to correctly receive them even under noisy channel conditions.  Protecting the header sufficiently will induce some overhead, but the assumption remains a reasonable one to make. First, the header is smaller than the message itself. Second, even in the routing case, the header and the state information need to be coded sufficiently. Third, the hashes $\mathbf{h_{I_i}}$ and $\mathbf{h_{x_i}}$ are contained within one hop. A node that receives $\mathbf{\underline{p_i}}=[\mathbf{a_i}, \mathbf{h_{I_i}}, \mathbf{h_{x_i}}, \mathbf{x_i}]$ does not need to repeat $\mathbf{h_{I_i}}$, only $\mathbf{h_{x_i}}$. Therefore, the overhead associated with the hashes is proportional to the in-degree of a node, and does not accumulate with the routing path length.

Assume that $v_i$ transmits $\mathbf{\underline{p_i}}=[\mathbf{a_i}, \mathbf{h_{I_i}}, \mathbf{h_{x_i}}, \mathbf{\hat{x}_i}]$, where $\mathbf{\hat{x}_i} = \mathbf{x_i} \oplus \mathbf{e}$,  $\mathbf{e} \in \{0,1\}^n$. If $v_i$ is misbehaving, then $\mathbf{e} \ne 0$. Our goal is to detect with high probability when $\mathbf{e} \ne 0$. Even if $|\mathbf{e}|$ is small (\ie the hamming distance between $\mathbf{\hat{x}_i}$ and $\mathbf{x_i}$ is small), the algebraic interpretation of $\mathbf{\hat{x}_i}$ and $\mathbf{x_i}$ may differ significantly. For example, consider $n=4$, $\mathbf{\hat{x}_i} =[0000]$, and $\mathbf{x_i} = [1000]$. Then, $\mathbf{e} = [1000]$ and $|\mathbf{e}| = 1$. However, the algebraic interpretations of $\mathbf{\hat{x}_i}$ and $\mathbf{x_i}$ are 0 and 8, respectively. Thus, even a single bit flip can alter the message significantly.

\subsection{Threat Model}\label{sec:watchdog_threat}

We assume powerful adversaries, who can eavesdrop their neighbor's transmissions, has the power to inject or corrupt packets, and are computationally unbounded. Thus, the adversary will find $\mathbf{\hat{x}_i}$ that will allow its misbehavior to be undetected, if there is any such $\mathbf{\hat{x}_i}$. However, the adversary does not know the specific realization of the random errors introduced by the channels. We denote the rate at which an adversary injects error (\ie performs bit flips to the payload) to be $p_{adv}$. The adversaries' objective is to corrupt the information flow without being detected by other nodes.

Our goal is to detect probabilistically a malicious behavior that is beyond the channel noise, represented by $BSC(p_{ik})$. Note that the algebraic watchdog does not completely eliminate errors introduced by the adversaries; its objective is to limit the errors introduced by the adversaries to be at most that of the channel. Channel errors (or those introduced by adversaries below the channel noise level) can be corrected using appropriate error correction schemes, which will be necessary even without Byzantine adversaries in the network.

The notion that adversarial errors should sometimes be treated as channel noise has been introduced previously in \cite{milcom}. Under heavy attack,
 attacks should be treated with special attention; while under light attack, the attacks can be treated as noise and corrected using error-correction schemes. The results in this paper partially reiterate this idea.

\section{Two-hop network: An Example}\label{sec:watchdog_example}
Consider a network (or a small neighborhood of nodes in a larger network) with nodes $v_1, v_2, ... v_m,$ $v_{m+1}$, $v_{m+2}$. Nodes $v_i$, $i\in [1,m]$, want to transmit $x_i$ to $v_{m+2}$ via $v_{m+1}$. A single node $v_i$, $i\in [1,m]$, cannot check whether $v_{m+1}$ is misbehaving or not even if $v_i$ overhears $\mathbf{x_{m+1}}$, since without any information about $x_j$ for $j \in [1,m]$, ${x}_{m+1}$ is completely random to $v_i$. On the other hand, if $v_i$ knows $x_{m+1}$ and $x_j$ for all $j \in [1, m]$, then $v_i$ can verify that $v_{m+1}$ is behaving with certainty; however, this requires at least $m-1$ additional reliable transmissions to $v_i$.

\begin{figure}[tbp]
\begin{center}
\includegraphics[width=0.33\textwidth]{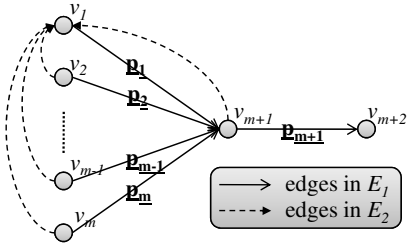}
\end{center}
\caption{A small neighborhood of a wireless network with $v_1$.}\label{fig:watchdog_general}\vspace*{-.3cm}
\end{figure}

We take advantage of the wireless setting, in which nodes can overhear their neighbors' transmissions. In Figure \ref{fig:watchdog_general},  we use the solid lines to represent the intended channels $E_1$, and dotted lines for the interference channels $E_2$ which we model with binary channels as mentioned in Section \ref{sec:watchdog_problemstatement}. Each node checks whether its neighbors are transmitting values that are consistent with the gathered information. If a node detects that its neighbor is misbehaving, then it can alert other nodes in the network and isolate the misbehaving node.

%


\begin{figure}[tbp]
\begin{center}
\includegraphics[width=0.3\textwidth]{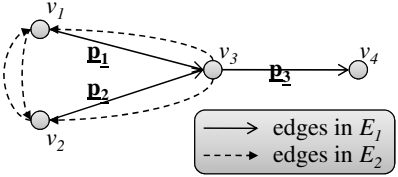}
\end{center}\vspace*{-.2cm}
\caption{A wireless network with $m=2$.}\label{fig:watchdog_2sources}\vspace*{-.3cm}
\end{figure}

In the next subsections, we shall use an example with $m=2$, as shown Figure \ref{fig:watchdog_2sources}. We introduce the graphical model which explains how a node $v_i$ checks its neighbor's behavior. Then, we use an algebraic approach to analyze and compute $\gamma$ and $\beta$ for this example network with $m=2$. In this section, we assume for simplicity that nodes do not code the payload -- \ie an error-correcting code of rate $R_i = 1$ is used.

Note that a malicious $v_3$ would not inject errors in $\mathbf{h_{x_3}}$ only, because the destination $v_4$ can easily verify if $\mathbf{h_{x_3}}$ is equal to $h(\mathbf{x_3})$. Therefore, $\mathbf{h_{x_3}}$ and $\mathbf{x_3}$ are consistent. In addition, $v_3$ would not inject errors in $\mathbf{h_{x_j}}$, $j \in I_3$, as each node $v_j$ can verify the hash of its message. On the other hand, a malicious $v_3$ can inject errors in $\mathbf{a_3}$, forcing $v_4$ to receive incorrect coefficients $\tilde{\alpha}_j$'s instead of $\alpha_j$'s. However, any error introduced in $\mathbf{a_3}$ can be translated to errors in $\mathbf{x_3}$ by assuming that $\tilde{\alpha}_j$'s are the correct coding coefficients. Therefore, we are concerned only with the case in which $v_3$ introduces errors in $\mathbf{x_3}$ (and therefore, in $\mathbf{h_{x_3}}$ such that $\mathbf{h_{x_3}} = h(\mathbf{x_3})$).

%

\subsection{Graphical model approach}\label{sec:watchdog_graphicalmodel}

We present a graphical approach to model the problem for $m=2$ systematically, and to explain how a node may check its neighbors. This approach may be advantageous as it lends easily to already existing graphical model algorithms as well as some approximation algorithms.

We shall consider the problem from $v_1$'s perspective. As shown in Figure \ref{fig:watchdog_trellis}, the graphical model has four layers: Layer 1 contains $2^{n+h}$ vertices, each representing a bit-representation of $[\mathbf{\tilde{x}_2, h(x_2)}]$; Layer 2 contains $2^n$ vertices, each representing a bit-representation of $\mathbf{x_2}$; Layer 3 contains $2^n$ vertices corresponding to $\mathbf{x_3}$; and Layer 4 contains $2^{n+h}$ vertices corresponding to  $[\mathbf{\tilde{x}_3, h(x_3)}]$. Edges exist between adjacent layers as follows:
\begin{itemize}
\item{\it{Layer 1 to Layer 2:}} An edge exists between a vertex $[\mathbf{v,u}]$ in Layer 1 and a vertex $\mathbf{w}$ in Layer 2 if and only if $\mathbf{h(w) = u}$. The edge weight is normalized such that the total weight of edges leaving $[\mathbf{v, u}]$ is 1, and the weight is proportional to:
     \[\mathbf{P}(\mathbf{v}| \text{ Channel statistics and } \mathbf{w} \text{ is the original message}),\] \vspace*{-.08cm}which is the probability that the inference channel outputs message $\mathbf{v}$ given an input message $\mathbf{w}$.
\item{\it{Layer 2 to Layer 3:}} The edges represent a permutation. A vertex $\mathbf{v}$ in Layer 2 is adjacent to a vertex $\mathbf{w}$ in Layer 3 if and only if $w = c+\alpha_2 v$, where $c = \alpha_1 x_1$ is a constant, $\mathbf{v}$ and $\mathbf{w}$ are the bit-representation of $v$ and $w$, respectively. The edge weights are all 1.
\item{\it{Layer 3 to Layer 4:}}  An edge exists between a vertex $\mathbf{v}$ in Layer 3 and a vertex $[\mathbf{w,u}]$ in Layer 4 if and only if $\mathbf{h(v) = u}$. The edge weight is normalized such that the total weight leaving $\mathbf{v}$ is 1, and is proportional to:
     \[\mathbf{P}(\mathbf{w}| \text{ Channel statistics and } \mathbf{v} \text{ is the original message}).\]
\end{itemize}

\begin{figure}[tbp]
\begin{center}
\includegraphics[width=0.45\textwidth]{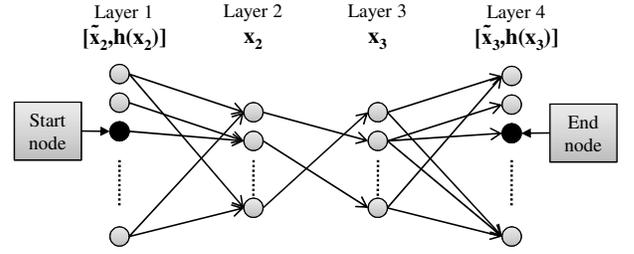}
\end{center}\vspace*{-.3cm}
\caption{A graphical model from $v_1$'s perspective}\label{fig:watchdog_trellis}\vspace*{-.3cm}
\end{figure}

Node $v_1$ overhears the transmissions from $v_2$ to $v_3$ and from $v_3$ to $v_4$; therefore, it receives $[\mathbf{\tilde{x}_2, h(x_2)}]$ and $[\mathbf{\tilde{x}_3, h(x_3)}]$, corresponding to the \emph{starting point} in Layer 1 and the \emph{destination point} in Layer 4 respectively. By computing the sum of the product of the weights of all possible paths between the starting and the destination points, $v_1$ computes the probability that $v_3$ is consistent with the information gathered.

This graphical model illustrates sequentially and visually the inference process $v_1$ executes. Furthermore, by using approximation algorithms and pruning algorithms, we may be able to simplify the computation as well as the structure of the graph. In addition, the graphical approach may be extend to larger networks, as we shall discuss in Section \ref{sec:watchdog_multisource}.

%
%
%

\subsection{Algebraic approach}\label{sec:watchdog_algebraicapproach}

We explain the inference process described above using the graphical model introduced in Section \ref{sec:watchdog_graphicalmodel}. Consider $v_1$. By assumption, $v_1$ correctly receives $\mathbf{a_2}$,  $\mathbf{a_3}$, $\mathbf{h_{I_2}}$, $\mathbf{h_{I_3}}$, $\mathbf{h_{x_2}}$, and $\mathbf{h_{x_3}}$. In addition, $v_1$ receives $\mathbf{\tilde{x}_2 = x_2 + e'}$ and $\mathbf{\tilde{x}_3 = x_3 + e''}$, where $\mathbf{e'}$ and $\mathbf{e''}$ are outcomes of the interference channels. Given $\mathbf{\tilde{x}_j}$ for $j = \{2,3\}$ and the transition probabilities, $v_1$ computes $r_{j\rightarrow 1}$ such that the sum of the probability that the interference channel from $v_j$ and $v_1$ outputs $\mathbf{\tilde{x}_j}$ given $\mathbf{x} \in B(\mathbf{\tilde{x}_j}, r_{j\rightarrow 1})$ is greater or equal to $1-\epsilon$ where $\epsilon$ is a constant, and $B(\mathbf{x}, r)$ is a $n$-dimensional ball of radius $r$ centered at $\mathbf{x}$. Now, $v_1$ computes $\tilde{X}_j = \{\mathbf{x}\ |\ h(x) = h(x_j)\} \cap B(\mathbf{\tilde{x}_j}, r_{j\rightarrow 1})$ for $j = \{2, 3\}$. Then, $v_1$ computes $\alpha_1 x_1 + \alpha_2 \hat{x}$ for all $\mathbf{\hat{x}} \in \tilde{X}_2$. Then, $v_1$ intersects $\tilde{X}_3$ and the computed $\alpha_1 x_1 + \alpha_2 \hat{x}$'s. If the intersection is empty, then $v_1$ claims that $R$ is misbehaving.

The set $\{\mathbf{x}\ |\ h(x) = h(x_2)\}$ represents the Layer 2 vertices reachable from the starting point ($[\mathbf{\tilde{x}_2, h(x_2)}]$ in Layer 1), and $\tilde{X}_2$ is a subset of the reachable Layer 2 vertices such that the total edge weight (which corresponds to the transition probability) from the starting point is greater than $1-\epsilon$. Then, computing $\alpha_1 x_1 + \alpha_2 \hat{x}$ represents the permutation from Layers 2 to 3. Finally, the intersection with $\tilde{X}_3$ represents finding a set of Layer 3 vertices such that they are adjacent to the destination point ($[\mathbf{\tilde{x}_3, h(x_3)}]$ in Layer 4) and their total transition probability to the destination point is greater than $1-\epsilon$.




\begin{lemma}
For $n$ sufficiently large, the probability of false detection, $\gamma \leq \epsilon$ for any arbitrary small constant $\epsilon$.
\end{lemma}
\begin{proof}
Assume that $v_3$ is not malicious, and transmits $\mathbf{x_3}$ and $\mathbf{h_{x_3}}$ consistent with $v_4$'s check. Then, for $n$ sufficiently large, $v_1$ can choose $r_{2\rightarrow 1}$ and $r_{3\rightarrow 1}$ such that the probability that the bit representation of $x_3 = \alpha_1 x_1 + \alpha_2 x_2$ is in $\tilde{X}_3$ and the probability that $\mathbf{x_2} \in \tilde{X}_2$ are greater than $1-\epsilon$. Therefore, $\tilde{X}_3 \cap \{\alpha_1 x_1 + \alpha_2 \hat{x}\ |\ \forall \mathbf{\hat{x}} \in \tilde{X}_2\} \ne \emptyset$ with probability arbitrary close to 1. Therefore, a well-behaving $v_3$ passes $v_1$'s check with probability at least $1-\epsilon$. Thus, $\gamma \leq \epsilon$.
\end{proof}

\begin{lemma}\label{thm:s1} The probability that a malicious $v_3$ is undetected from $v_1$'s perspective is given by
\[
\min\biggl\{1, \frac{\sum_{k=0}^{r_{1\rightarrow 2}} \binom{n}{k}}{2^{(h+n)}}\cdot \frac{\sum_{k=0}^{r_{2\rightarrow 1}} \binom{n}{k}}{2^{(h+n)}}\cdot \frac{\sum_{k=0}^{r_{3\rightarrow 1}} \binom{n}{k}}{2^{h}}\biggl\}.
\]
\end{lemma}
\begin{proof}
Assume that $v_3$ is malicious and injects errors into $\mathbf{x_3}$. Consider an element $\mathbf{z} \in \tilde{X}_3$, where $z = \alpha_1 x_1 + \alpha_2 x_2 + e = \alpha_1 x_1 + \alpha_2 (x_2 + e_2)$ for some $e$ and $e_2$. Note that, since we are using a field of size $2^n$, multiplying an element from the field by a randomly chosen constant has the effect of randomizing the product. Here, we consider two cases:

{\it Case 1:} If $x_2 + e_2 \notin \tilde{X}_2$, then $v_3$ fails $v_1$'s check.

{\it Case 2:} If $x_2 + e_2 \in \tilde{X}_2$, then $v_3$ passes $v_1$'s check; however, $v_3$ is unlikely to pass $v_2$'s check. This is because $\alpha_1 x_1 + \alpha_2 (x_2 + e_2) = \alpha_1 x_1 + \alpha_2 x_2 + \alpha_2 e_2 = \alpha_1 (x_1 + e_1) + \alpha_2 x_2$ for some $e_1$. Here, for uniformly random $\alpha_1$ and $\alpha_2$, $e_1$ is also uniformly random. Therefore, the probability that $v_3$ will pass is the probability that the uniformly random vector $x_1 + e_1$ belongs to $\tilde{X}_1 = \{x\ |\ h(x) = h(x_1)\} \cap B(\mathbf{\tilde{x}_1}, r_{1\rightarrow 2})$ where $v_2$ overhears $\mathbf{\tilde{x}_1}$ from $v_1$, and the probability that the interference channel from $v_1$ to $v_2$ outputs $\mathbf{\tilde{x}_1}$ given $\mathbf{x} \in B(\mathbf{\tilde{x}_1}, r_{1\rightarrow 2})$ is greater than $1-\epsilon$.
    \begin{align*}
    \mathbf{P}(\text{A malicious } v_3 \text{ passes $v_2$'s check})\\
    = \mathbf{P}(x_1 + e_1 \in \tilde{X}_1)
    = \frac{Vol(\tilde{X}_1)}{2^n},
    \end{align*}
    where $Vol(\cdot)$ is equal to the number of $\{0,1\}$-vectors in the given set. Since $Vol(B(x,r)) = \sum_{k=0}^r \binom{n}{k} \leq 2^n$, and the probability that $h(x)$ is equal to a given value is $\frac{1}{2^h}$, $Vol(\tilde{X}_1)$ is given as follows:
    \[
    Vol(\tilde{X}_1) = \frac{ Vol(B(\tilde{x}_1, r_{1\rightarrow 2}))}{2^h} = \frac{\sum_{k=0}^{r_1\rightarrow 2} \binom{n}{k}}{2^h}.
    \]%
%
%
Therefore, from $v_1$'s perspective, the probability that a $\mathbf{z} \in \tilde{X}_3$ passes the checks, $\mathbf{P}(\mathbf{z} \text{ passes check})$, is: \vspace*{-0.1cm}\[0\cdot \mathbf{P}(x_2 + e_2 \notin \tilde{X}_2) + \frac{\sum_{k=0}^{r_{1\rightarrow 2}} \binom{n}{k}}{2^{(h+n)}}\cdot \mathbf{P}(x_2 + e_2 \in \tilde{X}_2).\] Similarly, $\mathbf{P}(x_2 + e_2 \in \tilde{X}_2) = \frac{\sum_{k=0}^{r_{2\rightarrow 1}} \binom{n}{k}}{2^{(h+n)}}$, and $ Vol( \tilde{X}_3) = \frac{\sum_{k=0}^{r_{3\rightarrow 1}} \binom{n}{k}}{2^{h}}$. Then, the probability that $v_3$ is undetected from $v_1$'s perspective is the probability that \emph{at least one} $\mathbf{z} \in \tilde{X}_3$ passes the check:
\begin{align*}
\mathbf{P}(&\text{A malicious } v_3 \text{ is undetected from $v_1$'s perspective})\\
&= \min\{1, \mathbf{P}(\mathbf{z} \text{ passes check})\cdot Vol( \tilde{X}_3)\}.
\end{align*}
Note that $\mathbf{P}(\mathbf{z} \text{ passes check})\cdot Vol( \tilde{X}_3)$ is the expected number of $\mathbf{z} \in \tilde{X}_3$ that passes the check; thus, given a high enough $\mathbf{P}(\mathbf{z} \text{ passes check})$, would exceed 1. Therefore, we take $\min\{1, \mathbf{P}(\mathbf{z} \text{ passes check})\cdot Vol( \tilde{X}_3)\}$ to get a valid probability. This proves the statement.
\end{proof}

\begin{lemma}\label{thm:s2}
The probability that a malicious $v_3$ is undetected from $v_2$'s perspective is given by
\[
\min\biggl\{1, \frac{\sum_{k=0}^{r_{1\rightarrow 2}} \binom{n}{k}}{2^{(h+n)}}\cdot \frac{\sum_{k=0}^{r_{2\rightarrow 1}} \binom{n}{k}}{2^{(h+n)}}\cdot \frac{\sum_{k=0}^{r_{3\rightarrow 2}} \binom{n}{k}}{2^{h}}\biggl\},
\]
where $v_2$ overhears $\mathbf{\tilde{x}_3}$ from $v_3$, and the probability that the interference channel from $v_3$ to $v_2$ outputs $\mathbf{\tilde{x}_3}$ given $\mathbf{x} \in B(\mathbf{\tilde{x}_3}, r_{3\rightarrow 2})$ is greater than $1-\epsilon$.

\end{lemma}
\begin{proof}
By similar analysis as in proof of Lemma \ref{thm:s1}.
\end{proof}

\begin{theorem}\label{thm:main}
The probability of misdetection, $\beta$, is:
\[
\beta = \min\biggl\{1, \frac{\sum_{k=0}^{r_{1\rightarrow 2}} \binom{n}{k}}{2^{(h+n)}}\cdot \frac{\sum_{k=0}^{r_{2\rightarrow 1}} \binom{n}{k}}{2^{(h+n)}}\cdot \frac{1}{2^h} \sum_{k=0}^{r} \binom{n}{k}\biggl\},
\]
where $r = \min\{r_{3\rightarrow 1},r_{3\rightarrow 2}\}$.
\end{theorem}
\begin{proof}
The probability of misdetection is the minimum of the probability that $v_1$ and $v_2$ do not detect a malicious $v_3$. Therefore, by Lemma \ref{thm:s1} and \ref{thm:s2}, the statement is true.
\end{proof}

Theorem \ref{thm:main} shows that the probability of misdetection $\beta$ decreases with the hash size, as the hashes restrict the space of consistent codewords. In addition, since $r_{1\rightarrow 2}$, $r_{2\rightarrow 1}$, $r_{3\rightarrow 1}$, and $r_{3\rightarrow 2}$ represent the uncertainty introduced by the interference channels, $\beta$ increases with them. Lastly and the most interestingly, $\beta$ decreases with $n$, since $\sum^r_{k=0} \binom{n}{k} < 2^n$ for $r < n$. This is because network coding randomizes the messages over a field whose size is increasing exponentially with $n$, and this makes it difficult for an adversary to introduce errors without introducing inconsistencies.

We can apply Theorem \ref{thm:main} even when $v_1$ and $v_2$ cannot overhear each other. In this case, both $r_{1\rightarrow 2}$ and $r_{2 \rightarrow 1}$ equal to $n$, giving the probability of misdetection, $\beta = \min\{1, \sum_{k=0}^r \binom{n}{k}/8^h\}$ where $r = \min\{r_{3\rightarrow 1},r_{3\rightarrow 2}\}$. Here, $\beta$ highly depends on $h$, the size of the hash, as $v_1$ and $v_2$ are only using their own message and the overheard hashes.

The algebraic approach results in an analysis with exact formulae for $\gamma$ and $\beta$. In addition, these formulae are conditional probabilities; as a result, they hold regardless of a priori knowledge of whether $v_3$ is malicious or not. However, performing algebraic analysis is not very extensible with growing $m$.

\section{Algebraic Watchdog for Two-hop Network}\label{sec:watchdog_multisource}
We extend the algebraic watchdog to a more general two-hop network, as in Figure \ref{fig:watchdog_general}. We shall develop upon the trellis introduced in Section \ref{sec:watchdog_example}, and formally present a graphical representation of the inference process performed by a node performing algebraic watchdog on its downstream neighbor.

There are three main steps in performing the algebraic watchdog. First, we need to infer the original messages from the overheard information, which is captured by the transition matrix in Section \ref{sec:watchdog_transitionprob}. The second step consists of forming an opinion regarding what the next-hop node $v_{m+1}$ \emph{should} be sending, which is inferred using a trellis structure as shown in Section \ref{sec:watchdog_trellis} and a Viterbi-like algorithm in Section \ref{sec:watchdog_viterbi}. Finally, we combine the inferred information with what we overhear from $v_{m+1}$ to make a decision on how $v_{m+1}$ is behaving, as discussed in Section \ref{sec:watchdog_decision}. Figure \ref{fig:watchdog_inference} illustrates these three steps.

\begin{figure*}[tbp]
\begin{center}\hspace*{-.8cm}
\includegraphics[width=1.1\textwidth]{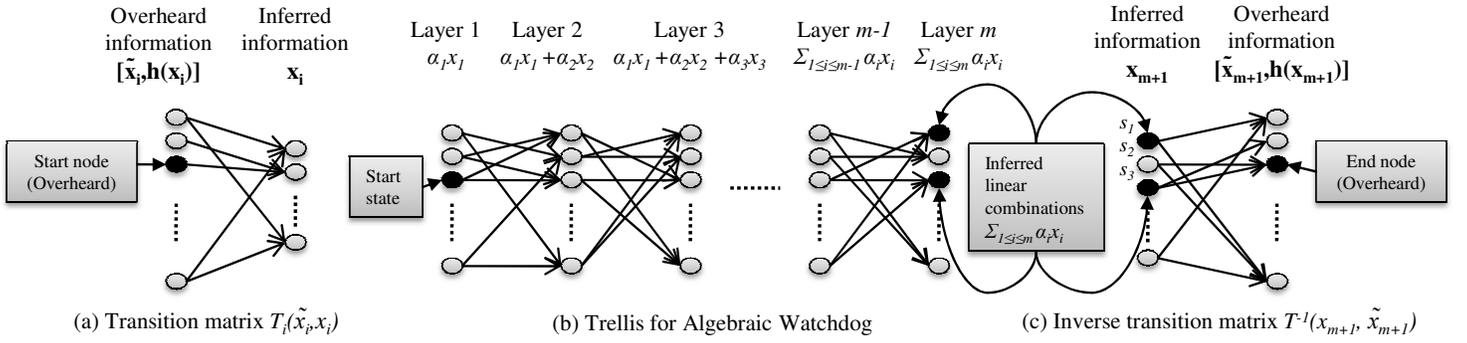}
\end{center}\vspace*{-.2cm}
\caption{Graphical representation of the inference process at node $v_1$. In the trellis, the transition probability from Layer $i-1$ to Layer $i$ is given by $T_i(\tilde{x}_i,x_i)$, which is shown in (a).}\vspace{-.3cm}\label{fig:watchdog_inference}\label{fig:watchdog_trellis}\label{fig:watchdog_transprob}\label{fig:watchdog_invtrans}
\end{figure*}

\subsection{Transition matrix}\label{sec:watchdog_transitionprob}

We define a \emph{transition matrix} $T_i$ to be a $2^{n(1-H(\frac{d_i}{n}))+\delta}\times 2^{n(1-H(\frac{d_i}{n}))}$ matrix, where $H(\cdot)$ is the entropy function.
\begin{align*}
T_i(\tilde{x}_i, y) &=
\begin{cases}
\frac{p_i(\tilde{x}_i, y)}{\mathcal{N}}, &\text{if $h(y) = h(x_i)$}\\
0, &\text{otherwise}
\end{cases},\\
p_i(\tilde{x}_i, y) &= p_{i1}^{\Delta(\mathbf{\tilde{x}_i, y})}(1-p_{i1})^{n-\Delta(\mathbf{\tilde{x}_i, y})},\\
\mathcal{N} &= \sum_{\{y|h(y)=h(x_i)\}} p_i(\tilde{x}_i, y),
\end{align*}
where $\Delta(\mathbf{x}, \mathbf{y})$ gives the Hamming distance between codewords $\mathbf{x}$ and $\mathbf{y}$. In other words, $v_1$ computes $\tilde{X}_i = \{x|h(x) = h(x_i)\}$ to be the list of \emph{candidates} of $x_i$. For any overheard pair $\mathbf{[\tilde{x}_i, h(x_i)]}$, there are multiple candidates of $x_i$ (\ie $|\tilde{X}_i|$) although the probabilities associated with each inferred $x_i$ are different. This is because there are uncertainties associated with the wireless medium, represented by $BSC(p_{i1})$.

For each $x\in \tilde{X}_i$, $p_i(\tilde{x}_i, x)$ gives the probability of $x$ being the original codeword sent by node $v_i$ given that $v_1$ overheard $\tilde{x}_i$ under $BSC(p_{i1})$. Since we are only considering $x\in \tilde{X}_i$, we normalize the probabilities using $\mathcal{N}$ to get the \emph{transition probability} $T_i(\tilde{x}_i, x)$. Note $T_i(\tilde{x}_i, y)=0$ if $h(y) \ne h(x_i)$.

The structure of $T_i$ heavily depends on the collisions of the hash function $h(\cdot)$ in use. Note that the structure of $T_i$ is independent of $i$, and therefore, a single transition matrix $T$ can be precomputed for all $i \in [1, m]$ given the hash function $h(\cdot)$. A graphical representation of $T$ is shown in Figure \ref{fig:watchdog_transprob}a. For simplicity of notation, we represent $T$ as a matrix; however, the transition probabilities can be computed efficiently using hash collision lists as well.



\subsection{Watchdog trellis}\label{sec:watchdog_trellis}

Node $v_1$ uses the information gathered to generate a trellis, which is used to infer the valid linear combination that $v_{m+1}$ should transmit to $v_{m+2}$. As shown in Figure \ref{fig:watchdog_trellis}b, the trellis has $m$ layers: each layer may contain up to $2^n$ states, each representing the inferred linear combination so far. For example, Layer $i$ consist of all possible values of $\sum_{j=1}^i \alpha_j x_j$.

The matrices $T_i, i\in[2,m]$, defines the connectivity of the trellis.
 Let $s_1$ and $s_2$ be states in Layer $i-1$ and Layer $i$, respectively. Then, an edge $(s_1, s_2)$ exists if and only if
\[
\exists\ x \text{ such that } s_1 + \alpha_i x = s_2,\ T_i(\tilde{x}_i, x) \ne 0.
\]
We denote $w_e(\cdot, \cdot)$ to be the edge weight,
where $w_e(s_1, s_2)=T_i(\tilde{x}_i, x)$ if edge $(s_1, s_2)$ exists, and zero otherwise.

\subsection{Viterbi-like algorithm}\label{sec:watchdog_viterbi}

We denote $w(s,i)$ to be the weight of state $s$ in Layer $i$. Node $v_1$ selects a \emph{start state} in Layer 1 corresponding to $\alpha_1 x_1$, as shown in Figure \ref{fig:watchdog_trellis}. The weight of Layer 1 state is $w(s, 1) = 1$ if $s = \alpha_1 x_1$, zero otherwise. For the subsequent layers, multiple paths can lead to a given state, and the algorithm keeps the aggregate probability of reaching that state. To be more precise, $w(s, i)$ is:
\[
w(s, i) = \sum_{\forall s' \in \text{Layer } i-1} w(s', i-1)\cdot w_e(s',s).
\]
By definition, $w(s,i)$ is equal to the total probability of $s = \sum_{j=1}^i \alpha_j x_j$ given the overheard information. Therefore, $w(s, m)$ gives the probability that $s$ is the valid linear combination that $v_{m+1}$ should transmit to $v_{m+2}$. It is important to note that $w(s,m)$ is dependent on the channel statistics, as well as the overheard information. For some states $s$, $w(s,m) =0$, which indicates that state $s$ can not be a valid linear combination; only those states $s$ with $w(s,m)>0$ are the \emph{inferred candidate linear combinations}.

The algorithm introduced above is a dynamic program, and is similar to the Viterbi algorithm. Therefore, tools developed for dynamic programming/Viterbi algorithm can be used to compute the probabilities efficiently.

\subsection{Decision making}\label{sec:watchdog_decision}

Node $v_1$ computes the probability that the overheard $\tilde{x}_{m+1}$ and $h(x_{m+1})$ are consistent with the inferred $w(\cdot, m)$ to make a decision regarding $v_{m+1}$'s behavior. To do so, $v_1$ constructs an \emph{inverse transition matrix} $T^{-1}$, which is a $2^{n(1-\frac{d_{m+1}}{n})} \times 2^{n(1-\frac{d_{m+1}}{n})+\delta}$ matrix whose elements are defined as follows:
\begin{align*}
T^{-1}(y, \tilde{x}_{m+1}) &=
\begin{cases}
\frac{p_{m+1}(\tilde{x}_{m+1}, y)}{\mathcal{M}}, &\text{if $h(y) = h(x_{m+1})$}\\
0, &\text{otherwise}
\end{cases},\\
\mathcal{M} &= \sum_{\{y|h(y)=h(x_{m+1})\}} p_{m+1}(\tilde{x}_{m+1}, y).
\end{align*}
Unlike $T$ introduced in Section \ref{sec:watchdog_transitionprob}, $T^{-1}(x, \tilde{x}_{m+1})$ gives the probability of overhearing $[\tilde{x}_{m+1}, h(x_{m+1})]$ given that $x \in \{y|h(y)=h(x_{m+1})\}$ is the original codeword sent by $v_{m+1}$ and the channel statistics. Note that $T^{-1}$ is identical to $T$ except for the normalizing factor $\mathcal{M}$. A graphical representation of $T^{-1}$ is shown in Figure \ref{fig:watchdog_invtrans}c.

In Figure \ref{fig:watchdog_invtrans}c, $s_1$ and $s_3$ are the inferred candidate linear combinations, \ie $w(s_1, m)\ne 0$ and $w(s_2,m)\ne 0$; the \emph{end node} indicates what node $v_1$ has overheard from $v_{m+1}$. Note that although $s_1$ is one of the inferred linear combinations, $s_1$ is not connected to the end node. This is because $h(s_1) \ne h(x_{m+1})$. On the other hand, $h(s_2) = h(x_{m+1})$; as a result, $s_2$ is connected to the end node although $w(s_2, m)= 0$. We define an inferred linear combination $s$ as \emph{matched} if $w(s,m)>0$ and $h(s) = h(x_{m+1})$.

Node $v_1$ uses $T^{-1}$ to compute the total probability $p^*$ of hearing $[\tilde{x}_{m+1}, h(x_{m+1})]$ given the inferred linear combinations by computing the following equation:
\[
p^* = \sum_{\forall s} w(s, m)\cdot T^{-1}(s, \tilde{x}_{m+1}).
\]
Probability $p^*$ is the probability of overhearing $\tilde{x}_{m+1}$ given the channel statistics; thus, measures the likelihood that $v_{m+1}$ is consistent with the information gathered by $v_1$. Node $v_1$ can use $p^*$ to make a decision on $v_{m+1}$'s behavior. For example, $v_1$ can use a threshold decision rule to decide whether $v_{m+1}$ is misbehaving or not: $v_1$ claims that $v_{m+1}$ is malicious if $p^* \leq t$ where $t$ is a threshold value determined by the given channel statistics; 
otherwise, $v_1$ claims $v_{m+1}$ is well-behaving.

Depending on the decision policy used, we can use the hypothesis testing framework to analyze the probability of false positive and false negative. Section \ref{sec:watchdog_example} provides such analysis for the simple two-hop network with a simple decision policy -- if the inferred linear combination and the message overheard from the next hop node is non-empty, we declare the node well-behaving. However, the main purpose of this paper is to propose a method in which we can compute $p^*$, which can be used to establish trust within a network. We note that it would be worthwhile to look into specific decision policies and their performance (\ie false positive/negative probabilities) as in \cite{algebraicwatchdog}.


%
%

\section{Analysis for Two-hop Network}\label{sec:watchdog_analysis}

We provide an analysis for the performance of algebraic watchdog for two-hop network.

\begin{theorem}\label{thm:number}
Consider a two-hop network as shown in Figure \ref{fig:watchdog_general}. Consider $v_j$, $j \in [1, m]$. Then, the number of \emph{matched} codewords is:
\[
2^{n\left[\sum_{i \ne j, i\in[1, m+1]} \left(H(p_{ij})-H(\frac{d_i}{n})\right) -1 \right] - m\delta}.
\]
\end{theorem}
\begin{proof}
Without loss of generality, we consider $v_1$. The proof uses on concepts and techniques developed for list-decoding \cite{listdecoding}. We first consider the overhearing of $v_k$'s transmission, $k \in [2, m]$. Node $v_1$ overhears $\tilde{x}_k$ from $v_k$. The noise introduced by the overhearing channel is characterized by $BSC(p_{k1})$; thus, $E[\Delta(\mathbf{x_k, \tilde{x}_k})] = np_{k1}$. Now, we consider the number of codewords that are within $B(\tilde{x}_k, np_{k1})$, the Hamming ball of radius $np_{k1}$ centered at $\tilde{x}_k$ is
$|B(\tilde{x}_k, np_{k1})| = 2^{n(H(p_{k1})-H(\frac{d_k}{n}))}$.
Node $v_1$ overhears the hash $h(x_k)$; thus, the number of codewords that $v_1$ considers is reduced to $2^{n(H(p_{k1})-H(\frac{d_k}{n}))-\delta}$.
Using this information, $v_1$ computes the set of inferred linear combinations, \ie $s$ where $w(s, m) >0$. Note that $v_1$ knows precisely the values of $x_1$. Therefore, the number of inferred linear combinations is upper bounded by:
\begin{align}
\prod_{k\in [2, m]}&\left(2^{n\left(H(p_{k1})-H(\frac{d_k}{n})\right)-\delta}\right)
\\ &= 2^{n\left[\sum_{k\in [2,m]}\left( H(p_{k1}) - H(\frac{d_k}{n})\right)\right] - (m-1)\delta} \label{eq:num_inferred}
\end{align}
Due to the finite field operations, these inferred linear combinations are randomly distributed over the space $\{0,1\}^n$.

Now, we consider the overheard information, $\tilde{x}_{m+1}$ from the downstream node $v_{m+1}$. By similar analysis as above, we can derive that there are $2^{n(H(p_{m+1, 1})-H(\frac{d_{m+1}}{n}))-\delta}$
codewords in the hamming ball $B(\tilde{x}_{m+1}, np_{m+1,1})$ with hash value $h(x_{m+1})$. Thus, the probability that a randomly chosen codeword in the space of $\{0,1\}^n$ is in $B(\tilde{x}_{m+1}, np_{m+1,1}) \cap \{x| h(x) = h(x_{m+1})\}$ is give by
\begin{equation}\label{eq:num_relay}
\frac{2^{n(H(p_{m+1, 1})-H(\frac{d_{m+1}}{n}))-\delta}}{2^n}.
\end{equation}

Then, the expected number of \emph{matched} codewords is the product of Equations (\ref{eq:num_inferred}) and (\ref{eq:num_relay}).
\end{proof}

If we assume that the hash is of length $\delta= \varepsilon n$, then the statement in Theorem \ref{thm:number} is equal to:
\begin{equation}
2^{n\left[\sum_{i \ne j, i\in[1, m+1]} H(p_{ij})- \left( \sum_{i \ne j, i\in[1, m+1]} H(\frac{d_i}{n}) +1  + m\varepsilon \right)\right]}.
\end{equation}
This highlights the tradeoff between the quality of overhearing channel and the redundancy (introduced by
$\mathcal{C}_i$'s
and the hash $h$). If enough redundancy is introduced, then $\mathcal{C}_i$ and $h$ together form an error-correcting code for the overhearing channels; thus, allows exact decoding to a single matched codeword.


The analysis also shows how adversarial errors can be interpreted. Assume that $v_{m+1}$ wants to inject errors at rate $p_{adv}$. Then, node $v_1$, although has an overhearing $BSC(p_{m+1,1})$, effectively experiences an error rate of $p_{adv} + p_{m+1,1} - p_{adv} \cdot p_{m+1, 1}$. Note that this does not change the set of the inferred linear combinations; but it affects $\tilde{x}_{m+1}$. Thus, overall, adversarial errors affect the set of matched codewords and the distribution of $p^*$. As we shall see in Section \ref{sec:watchdog_simulation}, the difference in distribution of $p^*$ between a well-behaving relay and adversarial relay can be used to detect malicious behavior.

\section{Protocol for Algebraic Watchdog}\label{sec:watchdog_multihop}

We use the two-hop algebraic watchdog from Section \ref{sec:watchdog_multisource} in a hop-by-hop manner to ensure a globally secure network. In Algorithm \ref{alg:protocol}, we present a distributed algorithm for nodes to secure the their local neighborhood. Each node $v$ transmits/receives data as scheduled; however, node $v$ randomly chooses to check its neighborhood, at which point node $v$ listens to neighbors transmissions to perform the two-hop algebraic watchdog from Section \ref{sec:watchdog_multisource}.

\begin{algorithm}[tbp]
\ForEach{node $v$}{
According to the schedule, transmit and receive data;
\If{$v$ decides to check its neighborhood}
{
Listen to neighbors' transmissions;\\
\ForEach{downstream neighbor $v'$}
{
Perform Two-hop Algebraic Watchdog on $v'$;
}}}
\caption{Distributed algebraic watchdog at $v$.}\label{alg:protocol}
\end{algorithm}

\begin{corollary}\label{thm:v}
Consider $v_{m+1}$ as shown in Figure \ref{fig:watchdog_general}. Assume that the downstream node $v_{m+2}$ is well-behaving, and thus, forces
$\mathbf{h_{x_{m+1}}} = h(x_{m+1})$. Let $\mathbf{\underline{p_i}}$ be the packet received by $v_{m+1}$ from parent node $v_i \in P(v)$. Then, if there exists at least one well-behaving parent $v_j \in P(v)$, $v_{m+1}$ cannot inject errors beyond the overhearing channel noise ($p_{m+1,j}$) without being detected.
\end{corollary}

Section \ref{sec:watchdog_analysis} noted that presence of adversarial error (at a rate above the channel noise) can be detected by a change in distribution of $p^*$. Corollary \ref{thm:v} does not make any assumptions on whether packets $\mathbf{\underline{p_i}}$'s are valid or not. Instead, the claim states that $v_{m+1}$ transmits a valid packet \emph{given} the packets $\mathbf{\underline{p_i}}$ it has received.


\begin{corollary}\label{thm:v2} Node $v$ can inject errors beyond the channel noise only if either of the two conditions are satisfied:
\begin{enumerate}
\item All its parent nodes $P(v) = \{u|(u,v)\in E_1\}$ are colluding Byzantine nodes;
\item All its downstream nodes, \ie receivers of the transmission $\mathbf{\underline{p_i}}$, are colluding Byzantine nodes.
\end{enumerate}
\end{corollary}

{\it Remark:} In Case 1), $v$ is not responsible to any well-behaving nodes. Node $v$ can transmit any packet without the risk of being detected by any well-behaving parent node. However, then, the min-cut to $v$ is dominated by adversaries, and the information flow through $v$ is completely compromised -- regardless of whether $v$ is malicious or not.

In Case 2), $v$ can generate any hash value since its downstream nodes are colluding adversaries. Thus, it is not liable to transmit a consistent hash, which is necessary for $v$'s parent nodes to monitor $v$'s behavior. However, note that $v$ is not responsible in delivering any data to a well-behaving node. Even if $v$ were well-behaving, it cannot reach any well-behaving node without going through a malicious node in the next hop. Thus, the information flow through $v$ is again completely compromised.

Therefore, Corollary \ref{thm:v2} shows that the algebraic watchdog can aid in ensuring correct delivery of data
 when the following assumption holds: for every intermediate node $v$ in the path between source to destination, $v$ has at least one well-behaving parent and at least one well-behaving child -- \ie there exists at least a path of well-behaving nodes. This is not a trivial result as we are not only considering a single-path network, but also multi-hop, multi-path network.

\section{Simulations}\label{sec:watchdog_simulation}

We present MATLAB simulation results that show the difference in distribution of $p^*$ between the well-behaving and adversarial relay. We consider a setup in Figure \ref{fig:watchdog_general}. We set all $p_{i1}$, $i \in [2, m]$ to be equal, and we denote this probability as $p_s = p_{i1}$ for all $i$. We denote $p_{adv}$ to be the probability at which the adversary injects error; thus, the effective error that $v_1$ observes from an adversarial relay is combined effect of $p_{m+1, 1}$ and $p_{adv}$. The hash function $h(x) = ax + b \mod 2^\delta$ is randomly chosen over $a, b \in \mathbf{F}_{2}^\delta$.

We set $n = 10$; thus, the coding field size is $2^{10}$. A typical packet can have a few hundreds to tens of thousand bits. Thus, a network coded packet with $n=10$ could have a few tens to a few thousands of symbols over which to perform algebraic watchdog. It may be desirable to randomize which symbols a node performs algebraic watchdog on, or when to perform algebraic watchdog. This choice depends not only on the security requirement, but also on the computational and energy budget of the node.

For each set of parameters, we randomly generate symbols from $\mathbf{F}_{2^{10}}$ ($n=10$ bits) and run algebraic watchdog. For each symbol, under a non-adversarial setting, we assume that only channel randomly injects bit errors to the symbol; under adversarial setting, both the channel and the adversary randomly inject bit errors to the symbol. For each set of parameters, we run the algebraic watchdog 1000 times. Thus, this is equivalent to running the algebraic watchdog on a moderately-sized packet (10,000 bits) or over several smaller packets, which are network coded using field size of $\mathbf{F}_{2^{10}}$.

For simplicity, nodes in the simulation do not use error-correcting codes; thus, $d_i = 0$ for all $i$. This limits the power of the algebraic watchdog; thus, the results shown can be further improved by using error correcting codes $\mathcal{C}_i$.

We denote $p^*_{adv}$ and $p^*_{relay}$ as the value of $p^*$ when the relay is adversarial and is well-behaving, respectively. We denote $var_{adv}$ and $var_{relay}$ to be the variance of $p^*_{adv}$ and $p^*_{relay}$. We shall show results that show the difference in distribution of $p^*_{adv}$ and $p^*_{relay}$ from $v_1$'s perspective. Note that this illustrates that only \emph{one} good parent node, \ie $v_1$ in our simulations, is sufficient to notice the difference in distribution of $p^*_{adv}$ and $p^*_{relay}$. Thus, confirming our analysis in Section \ref{sec:watchdog_multihop}. With more parent nodes performing the check independently, we can improve the probability of detection.


Our simulation results coincide with our analysis and intuition. Figure \ref{fig:matlab_padv} shows that adversarial above the channel noise can be detected. First of all, for all values of $p_{adv} > 0$, $p^*_{adv} < p^*_{relay}$; thus, showing that adversarial errors can be detected. Furthermore, the larger the adversarial error injection rate, the bigger the difference in the distributions of $p^*_{adv}$ and $p^*_{relay}$. When adversarial error rate is small, then the effective error $v_1$ sees in the packet can easily be construed as that of the channel noise, and thus, if appropriate channel error correcting code is used, can be corrected by the downstream node. As a result, the values/distributions of $p^*_{relay}$ and $p^*_{adv}$ are similar. However, as the adversarial error rate increases, there is a divergence between the two distributions. Note that the difference in the distributions of $p^*_{relay}$ and $p^*_{adv}$ is not only in the average value. The variance $var_{relay}$ is relatively constant throughout ($var_{relay}$ is approximately 0.18 throughout). On the other hand, $var_{adv}$ generally decreases with increase in $p_{adv}$. For small $p_{adv}$, the variance $var_{adv}$ is approximately 0.18; while for large $p_{adv}$, the variance $var_{adv}$ is approximately $0.08$. This trend intuitively shows that, with increase in $p_{adv}$, not only do we detect that the adversarial relay more often (since the average value of $p^*_{adv}$ decreases), but we are more confident of the decision.

\begin{figure}[tbp]
\begin{center}
\includegraphics[width=0.46\textwidth]{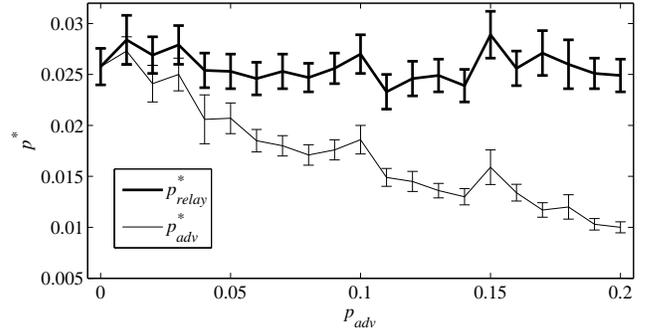}
\end{center}\vspace*{-.2cm}
\caption{The average value of $p^*$ with well-behaving relay (denoted $p^*_{relay}$) and adversarial relay (denoted $p^*_{adv}$) over 1000 random iterations of algebraic watchdog. The error bars represent the variance, $var_{relay}$ and $var_{adv}$. We set $m = 3$, $n = 10$, $\delta = 2$, and $p_s = p_{m+1, 1} = 10\%$. We vary $p_{adv}$, the adversary's error injection rate.}\label{fig:matlab_padv}
\end{figure}

Figure \ref{fig:matlab_delta} shows the affect of the size of the hash. With increase in redundancy (by using hash functions of length $\delta$), $v_1$ can detect malicious behavior better. This is true regardless of whether the relay is well-behaving or not. Node $v_1$'s ability to judge its downstream node increases with $\delta$. Thus, $p^*_{adv}$ for $\delta$ is generally higher than that of $\delta'$ where $\delta > \delta'$. This holds for $p^*_{relay}$ as well. However, for any fixed $\delta$, node $v_1$ can see a distinction between $p^*_{relay}$ and $p^*_{adv}$ as shown in Figure \ref{fig:matlab_delta}. A similar trend to that of Figure \ref{fig:matlab_padv} can be seen for  the distributions of $p^*_{adv}$ and $p^*_{relay}$ for each value of $\delta$. The interesting result is that even for $\delta = 0$, \ie we include no redundancy or hash, node $v_1$ is able to distinguish an adversarial relay from a well-behaving relay.

\begin{figure}[tbp]
\begin{center}
\includegraphics[width=0.5\textwidth]{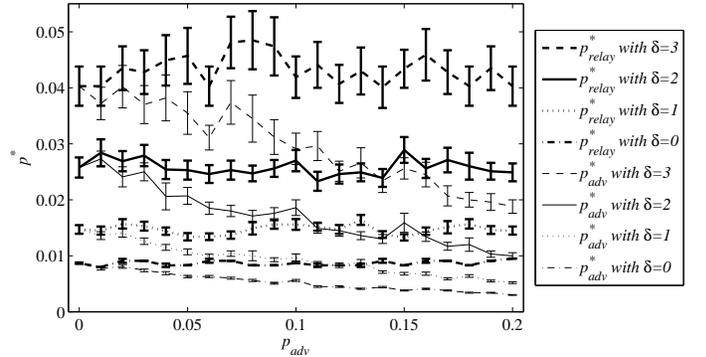}
\end{center}\vspace*{-.2cm}
\caption{The average value of $p^*$ with well-behaving relay (denoted $p^*_{relay}$) and adversarial relay (denoted $p^*_{adv}$) over 1000 random iterations of algebraic watchdog. We vary the value of $\delta$, the length of the hash function used, and $p_{adv}$, the adversary's error injection rate. The error bars represent the variance, $var_{relay}$ and $var_{adv}$. We set $m = 3$, $n = 10$, and $p_s = p_{m+1, 1} = 10\%$. }\label{fig:matlab_delta}\vspace*{-.3cm}
\end{figure}

Results in Figure \ref{fig:matlab_ps} confirms our intuition that the better $v_1$'s ability to collect information from $v_i$'s, $\i ne 1$, the better its detection ability. If node $v_1$ is able to infer better or overhear $x_i$ with little or no errors, the better its inference on what the relay node should be transmitting. Thus, as overhearing channel progressively worsens ($p_s$ increases), $v_1$'s ability to detect malicious behavior deteriorates; thus, unable to distinguish between a malicious and a well-behaving relay.

\begin{figure}
\begin{center}
\subfloat[Linear scale]{\includegraphics[width=0.24\textwidth]{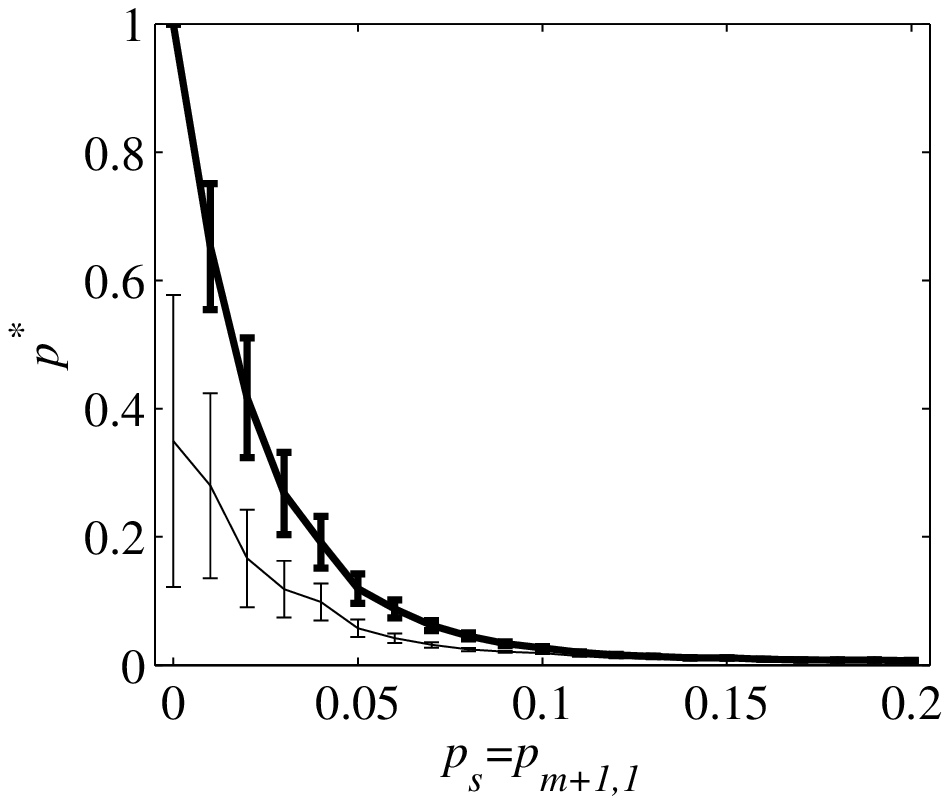}\label{fig:pslin}}
\subfloat[Log scale]{\includegraphics[width=0.242\textwidth]{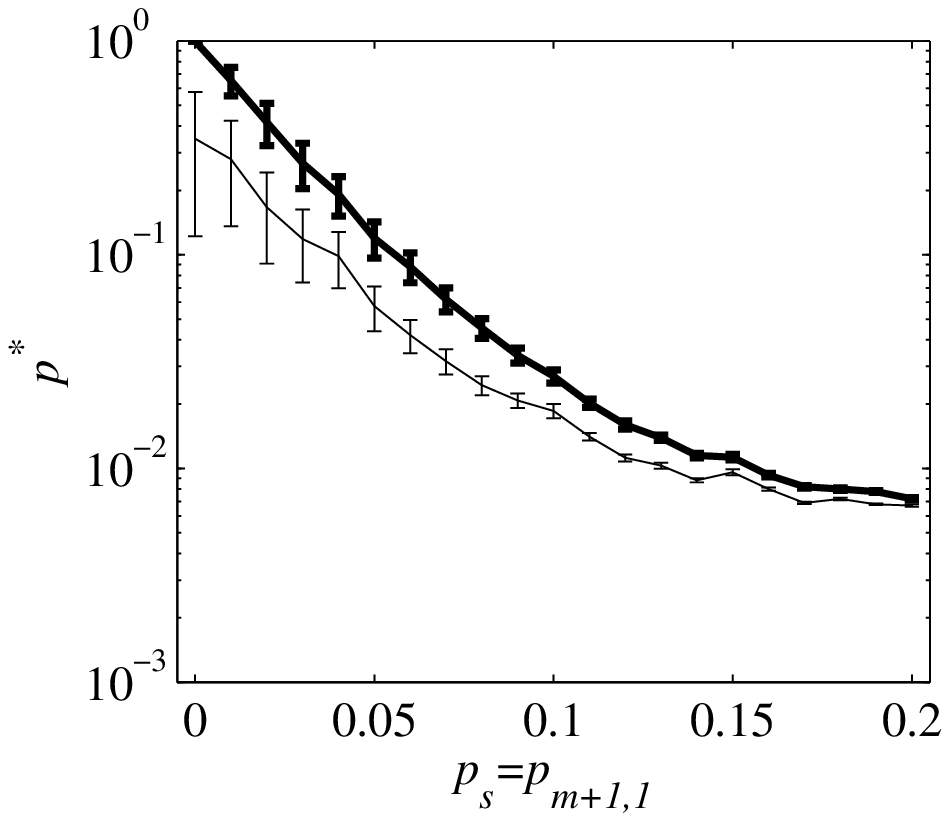}\label{fig:pslog}}
\end{center}\vspace*{-.3cm}\caption{The average value of $p^*$ with well-behaving relay (denoted $p^*_{relay}$) and adversarial relay (denoted $p^*_{adv}$) over 1000 random iterations of algebraic watchdog. We vary the value of $p_s = p_{m+1,1}$, the quality of overhearing channels. The error bars represent the variance, $var_{relay}$ and $var_{adv}$. We set $m = 3$, $n = 10$, and $p_{adv} = 10\%$.}\label{fig:matlab_ps}\vspace*{-.3cm}
\end{figure}

Finally, we note the effect of $m$, the number of nodes in the network, in Figure \ref{fig:matlab_m}. Node $v_1$'s ability to check $v_{m+1}$ is reduced with $m$. When $m$ increases, the number of messages $v_1$ has to infer increases, which increases the uncertainty within the system. However, it is important to note that as $m$ increases, there are more nodes $v_i$'s, $i \in [1, m]$ that can independently perform checks on $v_{m+1}$. This affect is not captured by the results shown in Figure \ref{fig:matlab_m}.

\begin{figure}
\begin{center}
\subfloat[Linear scale]{\includegraphics[width=0.24\textwidth]{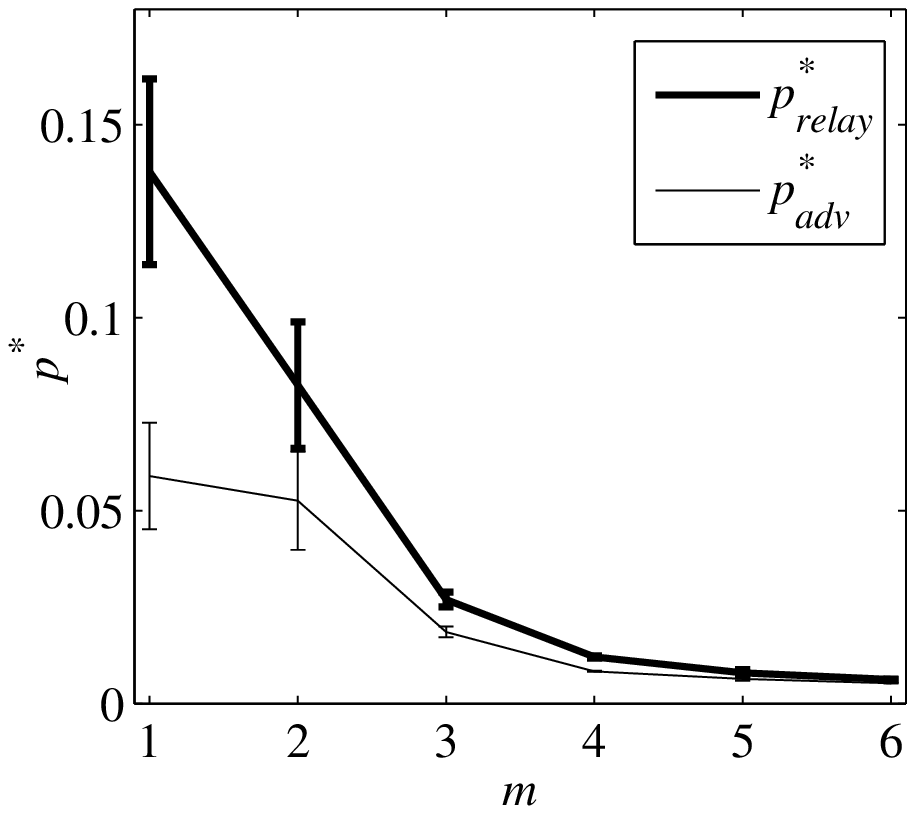}\label{fig:mlin}}
\subfloat[Log scale]{\includegraphics[width=0.231\textwidth]{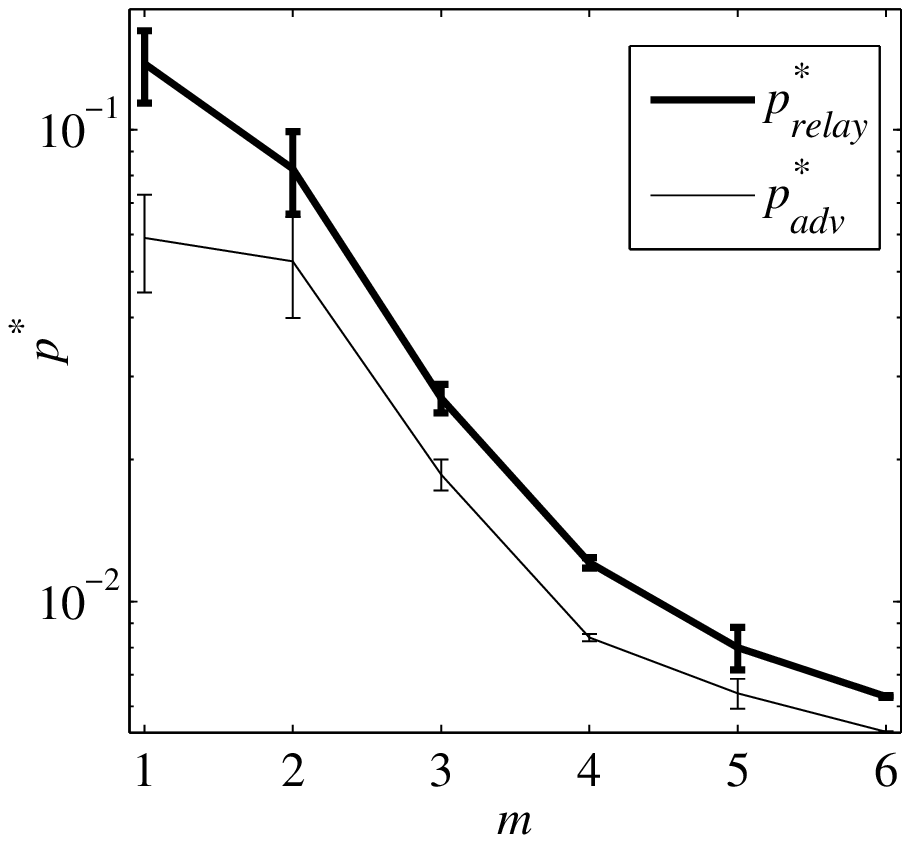}\label{fig:mlog}}
\end{center}\vspace*{-.3cm}\caption{The average value of $p^*$ with well-behaving relay (denoted $p^*_{relay}$) and adversarial relay (denoted $p^*_{adv}$) over 1000 random iterations of algebraic watchdog. We vary the value of $m$, the number of nodes using $v_{m+1}$ as a relay. The error bars represent the variance, $var_{relay}$ and $var_{adv}$. We set $m = 3$, $n = 10$, and $p_s = p_{m+1, 1}=p_{adv} = 10\%$.}\label{fig:matlab_m}\vspace*{-.3cm}
\end{figure} 

\section{Conclusions and Future Work}\label{sec:watchdog_conclusions}
We proposed the \emph{algebraic watchdog}, in which nodes can verify their neighbors probabilistically and police them locally by the means of overheard messages in a coded network. Using the algebraic watchdog scheme, nodes can compute a probability of consistency, $p^*$, which can be used to detect malicious behavior. Once a node has been identified as malicious, these nodes can either be punished, eliminated, or excluded from the network by using reputation based schemes such as \cite{marti}\cite{reputation}.

We first presented a graphical model and an analysis of the algebraic watchdog for two-hop networks. We then extended the algebraic watchdog to multi-hop, multi-source networks. We provided a trellis-like graphical model for the detection inference process, and an algorithm that may be used to compute the probability that a downstream node is consistent with the overheard information. We analytically showed how the size of the hash function, minimum distance of the error-correcting code used, as well as the quality of the overhearing channel can affect the probability of detection. Finally, we presented simulation results that support our analysis and intuition.
%
%

Our ultimate goal is to design a network in which the participants check their neighborhood locally to enable a secure global network –- \ie a self-checking network. Possible future work includes developing inference methods and approximation algorithms to decide efficiently aggregate local trust information to a global trust state.

\bibliographystyle{IEEEtran}
\bibliography{main_watchdog}

\end{document}